\documentclass[10pt,letterpaper,twocolumn,conference]{IEEEtran}

\IEEEoverridecommandlockouts

\usepackage[T1]{fontenc}
\usepackage[latin1]{inputenc}
 \usepackage{verbatim,multirow}
 \usepackage{float}
 \usepackage{amsmath}
\usepackage{subfigure}
 \usepackage{graphicx}
 \usepackage{amssymb}
\usepackage{amsthm}   
\usepackage{bm}
 \usepackage{calc}
 \makeatletter

\usepackage{gastex}

 \floatstyle{ruled}
 \newfloat{algorithm}{tbp}{loa}
 \floatname{algorithm}{Algorithm}

 \usepackage{cite}
 \usepackage{hhline}
 \usepackage{subfigure}
 \usepackage{array}
 \usepackage{acronym}
\usepackage{color}
\usepackage{algorithm}
\usepackage{algorithmic}
\usepackage{float}
\newtheorem{lem}{Lemma}
\theoremstyle{remark}
\newtheorem{rem}{Remark}

\newtheorem{definition}{Definition}

 \makeatletter \newcommand \listoftodos{\chapter{Todo list} \@starttoc{todocontents}}
 \newcommand\l@todo[2]
     {\par\noindent pg \textit{#2}: \parbox{10cm}{#1}\par} \makeatother

\acrodef{RV}{random variable}

\acrodef{i.i.d.}{independent, identically distributed}

\acrodef{PDF}{probability distribution function}
\acrodef{PMF}{probability mass function}
\acrodef{CDF}{cumulative distribution function}
\acrodef{HF}{high frequency}
\acrodef{ch.f.}{characteristic function}

\acrodef{NPV}{net present value}
\acrodef{PV}{present value}
\acrodef{CF}{cash flow}
\acrodef{DCF}{discounted cash flow}
\acrodef{WACC}{weighted average cost of capital}

\acrodef{CDN}{content distribution network}
\acrodef{NC}{network coding}
\acrodef{DC}{data center}
\acrodef{PUE}{power usage effectiveness}
\acrodef{RLC}{random linear coding}
\acrodef{ACPI}{advanced configuration and power interface}
\acrodef{RLNC}{random linear network coding}
\acrodef{SAN}{storage area network}
\acrodef{NAS}{network attached storage}
\acrodef{NCS}{network coded storage}
\acrodef{UCS}{uncoded storage}
\acrodef{AWS}{Amazon Web Services}
\acrodef{HD}{high definition}
\acrodef{HLS}{HTTP Live Streaming}
\acrodef{NCC}{no coefficient-cycling}
\acrodef{CC}{coefficient-cycling}
\acrodef{HDD}{hard disk drive}
\acrodef{SSD}{solid state drive}
\acrodef{IDNC}{instantly decodable network coding}
\acrodef{SFM}{state feedback matrix}
\acrodef{DSM}{drive state matrix}
\acrodef{d.o.f.}{degree of freedom}
\acrodef{PMP}{point-to-multipoint}
\acrodef{MCMC}{Markov-Chain-Monte-Carlo}

\author{Ulric J.~Ferner$^{\dag}$, Parastoo Sadeghi$^{\ddag}$, Neda Aboutorab$^{\ddag}$, and Muriel M\'edard$^{\dag}$\\ $^{\dag}$Research Laboratory for Electronics,  Massachusetts Institute of Technology, Cambridge,
    MA 02139, USA \\$^{\ddag}$Research School of Engineering, Australian National University, Canberra ACT 0200, Australia}
\title{Scheduling Advantages of Network Coded Storage in Point-to-Multipoint Networks}
\date{\today}

\begin{document}
\maketitle
\begin{abstract}
We consider scheduling strategies for \ac{PMP} \acp{SAN} that use \ac{NCS}.  In particular, we present a
simple \ac{SAN} system model, two server scheduling algorithms for \ac{PMP} networks, and analytical
expressions for internal and external blocking probability.  We point to select
scheduling advantages in \ac{NCS} systems under
\emph{normal} operating conditions, where content requests can be temporarily denied owing to
finite system capacity from drive I/O
access or storage redundancy limitations.  
 \ac{NCS} can lead to improvements in throughput and blocking probability due to
  increased immediate scheduling options, and complements other well documented \ac{NCS} advantages
  such as regeneration,
  and can be used as a guide for future storage system design.  
\end{abstract}
\section{Introduction}
\label{sec:intro}

The prolific growth of online content and streaming video makes serving content
requests to multiple users simultaneously an important technique for modern
storage area networks (SANs).  Two fundamental measures of service quality are system external blocking probability, i.e.,
the probability that a requesting  user is denied immediate access to content, as well as system throughput.
Under \emph{normal} operating conditions and given perfect scheduling, \acf{NCS} has been identified as a
promising  technique to reduce
blocking probability.  For instance,\cite{FerMedSol:12} used queuing theory
to show that network coding can reduce system blocking probability.  In
this paper we build upon this idea and develop simple and intuitive server scheduling algorithms
for such  \ac{NCS} systems.   We then
explore their impact on both throughput as well as blocking  probability.  The main contributions of this paper are:
\begin{itemize}
\item We introduce a simple storage model for \acf{PMP} storage networks that allows direct evaluation of blocking probability and system throughput;
\item  Using this model, we propose two intuitive scheduling algorithms---one for \acf{UCS} and one
for \ac{NCS}---that can achieve maximal throughput;
\item We quantify the blocking probability and throughput savings of \ac{NCS} over \ac{UCS}
scheduling, showing that a small improvement in throughput translates to a comparatively
large improvement in blocking probability.  
\end{itemize}

This paper builds upon and complements existing work in this area.  The use of  \ac{NCS} as regenerating codes
is a well studied repair technique to enhance \ac{SAN} reliability
\cite{DimRamWuSuh:11} in both centralized and distributed systems.  This particularly holds in \emph{less common} operating
conditions, such as permanent  drive failures.  In modern systems \emph{traffic-induced} temporary unavailability significantly dominates
disk failures \cite{ForLabPopStoTruBarGriQui:10}, and so like in \cite{FerMedSol:12}, this paper focusses on \emph{normal} operating
conditions and seeks to avoid highly transient and temporary bottlenecks in data liveness.  General scheduling for coded storage in point-to-point networks, when users are served sequentially instead
of simultaneously, are considered in \cite{ShaLeeRam:12,HuaPawHasRam:12}.  

Server scheduling is also well studied in matched networks such as cross-bar switches.  Throughput-optimal
schedules are considered for $ N\times M$ point-to-point cross-bar switches using graph theory and
techniques such as the Birkhoff-von Neumann theorem \cite{AndOwiSaxTha:93}.  
Switches with multicast and broadcast capabilities with
a queueing analysis flavor are considered in \cite{MarBiaGiaLeoNer:03}.  References \cite{PraMcKAhu:97,YuReuBer:11} attempt to map the
multicast problem in cross-bar switches to simpler problems such as block-packing
games and round-robin based multicast. By characterizing flow conflict graphs and their corresponding stable set polytopes in multicast cross-bar switches, \cite{KimSunMedEryKot:11} proposed online and offline network coding schedules for enhancing throughput. For general PMP storage networks, developing appropriate storage models, corresponding conflict graphs, and throughput optimal scheduling is an interesting and largely unaddressed area of research. This paper takes a first step towards this by considering a particular kind of \ac{PMP} network, namely broadcast, and by developing intuitive coded and uncoded leader-based scheduling, which do not explicitly require conflict graph construction. Chunk scheduling problems in uncoded peer-to-peer
networks, as opposed to \acp{PMP}, are considered in \cite{FenLi:11}, and for star-based broadcast
networks in \cite{RouSiv:97}. Note also that unlike classical asynchronous broadcast problems \cite{AskFraZdo:01,Hu:01}, our goal is not to reduce content delivery delay or
to optimize caching.  Instead, by taking into account  intermittent drive availability, we aim to determine the impact
of  scheduling drive reads and the impact of content storage format on blocking probability and throughout.  We expect that by using appropriate caching, system
performance can be further improved.  However, this is beyond the scope of this work.  

The remainder of this paper is organized as follows.  Section \ref{sec:model} details our system
model.  Section \ref{sec:schemes} presents service schemes and Section \ref{sec:numResults}
describes numerical results.  Section \ref{sec:conc} concludes the paper.



\section{System Model}
\label{sec:model}
Fig.~\ref{fig:systemModel} depicts our tree-structured
connectivity model made of single server $S$,
connected to $R$ drives, that receives user requests for content. 

\begin{figure}[tb]
  \centering
  \includegraphics[width=\linewidth]{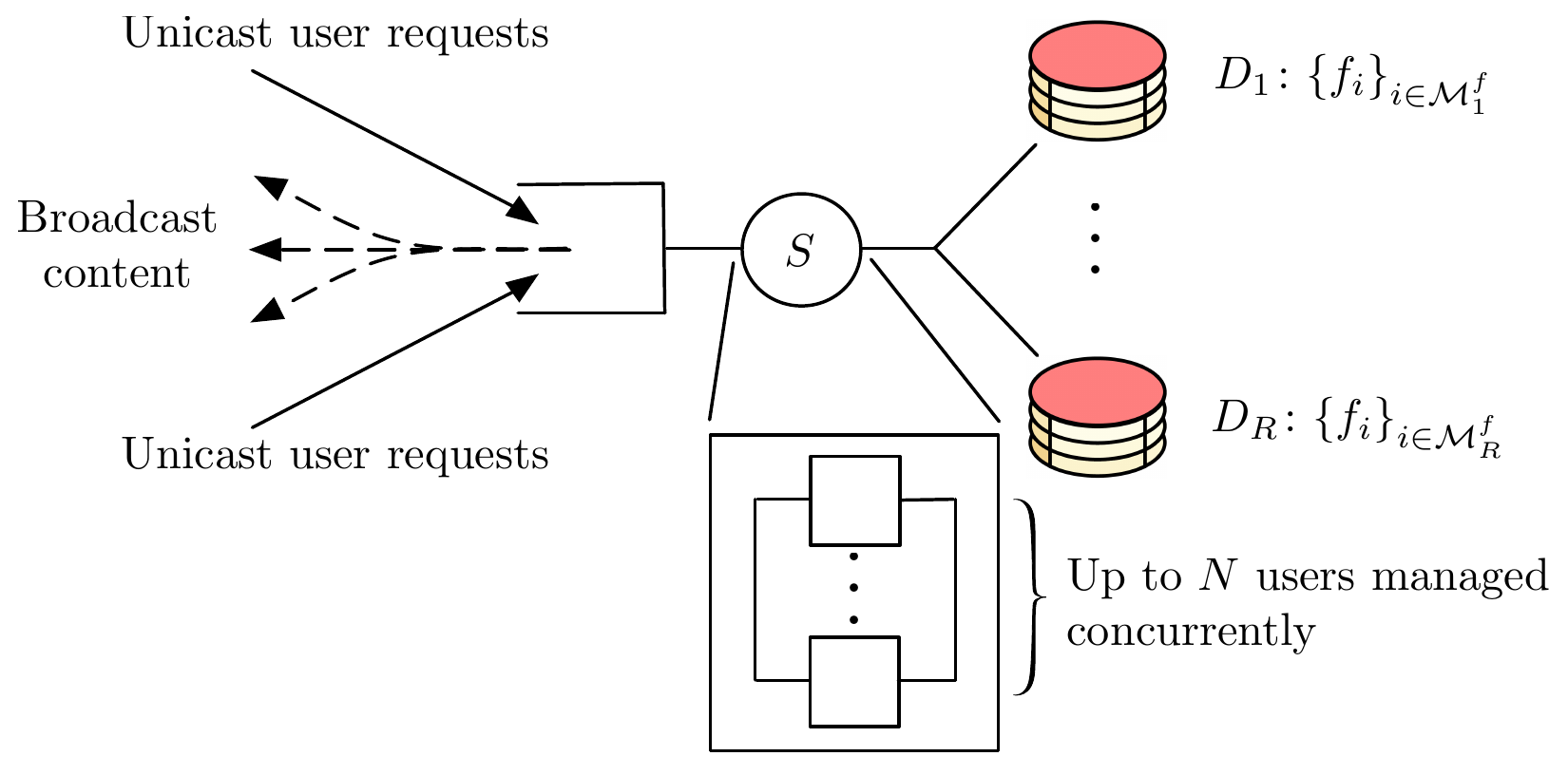}
  \caption{System model.}
  \label{fig:systemModel}\vspace{-4mm}
\end{figure}

\vspace{-0.5cm}
\subsection{Drives}
\label{sec:drives}

The \ac{SAN} in Fig.~\ref{fig:systemModel} stores a single chunked file $ \mathcal F = \{f_1,f_2,\cdots,f_T\} $, where $T$
is the number of chunks in $ \mathcal F$, and $ \mathcal F$ is stored across a set of drives
$\mathcal{D} = \{D_1, \cdots, D_R\}$.  If drive $ D_{i}$ receives a read request for chunk $ f_{j}
$, and if $D_{i}$ stores $ f_{j}$ and $ D_i$ is available, then it takes one timeslot to read out that content and broadcast to all users.   We model the overall effect of drives having
finite I/O access bandwidth with parameter $
P_{b}^{D}$, where $ P_{b}^{D}$ is the probability that any drive is blocked in timeslot $t$.  For
simplicity, we assume drives  are blocked
independently of one another and across timeslots.\footnote{This blocking model can be applicable where other servers have access to the same drives and therefore, there is some probabilistic traffic-induced blocking observed by $S$. More realistic models for traffic-induced drive blocking as well as more general PMP traffic patterns are beyond the scope of this paper and subject of our current research. }

In \ac{UCS}, let $\mathcal M^d_{i} \subseteq \mathcal {D}$ be the collection of drives that hold
uncoded file chunk $f_i$ and conversely, let $\mathcal M^f_i \subset \mathcal F$ be the collection
of file chunks held by drive $D_i$.  The only requirement of chunks to drives is that $R$ drives collectively
hold at least one copy of the whole file, i.e., $\mathcal F = \cup_{i = 1}^{R} \mathcal M^f_i$.

In \ac{NCS},  the $r$th coded file chunk is represented as \cite{FerMedSol:12}
\begin{equation}
  c_{r} = \sum_{j = 1}^{T} \alpha_{j,r} f_j
\end{equation}
where $\alpha_{j,r}$ is the encoding coefficient of file chunk $f_j$ and the corresponding encoding vector is
\begin{equation}
  \mathbf{k}_{r} = \sum_{j = 1}^{T} \alpha_{j,r} \mathbf{e}_j\label{eq:encoding:vec}
\end{equation}
In \eqref{eq:encoding:vec}, $\mathbf{e}_j = [e_{j,1}, \cdots, e_{j,T}]$ is the unit encoding row vector of length $T$
with elements $e_{j,r} = \delta_{j,r}$.  Function $\delta_{j,r}$ is the Kronecker delta function
with $\delta_{j,r} = 1$ iff $j=r$.  We assume that a total of $H$ linearly coded chunks $c_1$ to
$c_H$ are stored onto drives via some MDS code, such that any $T$ coded chunks are
linearly independent so that the original file chunks can be recovered from them using Gaussian
elimination.  If encoding coefficients $\alpha_{j,r}$ are randomly selected from a finite field
$\mathbb F_q$ with sufficiently large size $q$, this requirement is satisfied with high
probability \cite{HoMedKoeEffShiKar:06}.
 
\vspace{-0.1cm}
\subsection{Server}
\label{sec:servers}

We assume server $S$  has a bounded buffer to manage concurrent user requests.  Let $N$ be the
maximum number of users that can be managed and serviced concurrently and suppose $S$  operates in
slotted time.  In particular, in any timeslot, $S$ can serve at most $N$ active requests for content.   A user request for a content is cleared
from the buffer when all its requested file chunks have been transmitted by $S$.  

Any additional request beyond $N$ for the same content will be \emph{externally blocked}.\footnote{$N$ is an
  arbitrary, possibly time varying, quantity and hence this model does not limit our analysis. }
We will discuss the relation between external and \emph{internal} blocking in Section \ref{sec:metrics}. This is a similar model to existing drive blocking models \cite{FerMedSol:12} and existing practical
server experimentation test \cite{FerLonPedVolMed:13}.  When a user request arrives
and is not externally blocked, one slot of
the server buffer is allocated to manage and service this user request.   We make the following
additional assumptions about how $S$ retrieves content from $ \mathcal D$:
\begin{itemize}
    \item Let the vector $\bm{b}(t)$ of size $R$ be the drive availability vector, where $b_i(t) = 0$
means drive $D_i$ is free for reads and $b_i(t) = 1$ means it is busy in timeslot $t$. We assume that $\bm{b}(t)$ can be obtained by the server at the beginning of timeslot $t$ with negligible time overhead.
      \item In timeslot $t$, based on $ \mathbf{b}(t)$,  $S$ can choose to send a read request to
  access a single drive and read a single chunk.
    \item At the end of timeslot $t$, $S$  broadcasts the received chunk $ x(t)$ to users active in the
buffer.  
\end{itemize}
We assume perfect communication so when $S$ broadcasts content all active users receive that
content without error.  

\vspace{-0.1cm}
\subsection{Users}
\label{sec:users}
 We model users with the following key parameters:
\begin{itemize}
      \item User requests arrives at $S$ following a Poisson process with rate $ \lambda$.
      \item All user requests are for the entire file $ \mathcal F$, so in the long-term there is
  uniform traffic demand across file chunks. 
      \item Users currently being managed and serviced by $S$ are referred to as \textit{active
    users}, which we denote by $\mathcal U_{A}$, which is a subset of all serviceable users $\mathcal U = \{u_m\}$. 
\end{itemize}
Each user $u_m$ stores the received encoding vectors up to timeslot $t$ in a buffer (matrix)
denoted by $\mathbf{K}_m(t)$. This is called the knowledge space of user $u_m$ at timeslot $t$.  The rank of knowledge space of user $u_m$ at timeslot $t$  is denoted by $r_m(t)= \mathtt {rank}(\mathbf{K}_m(t))$.

A user is said to \emph{receive a new \ac{d.o.f.}} if the rank of its knowledge space increases by
one after reception of a chunk from $S$, that is, if
 $ r_m(t+1)= r_m(t)+1 \, .$
A file chunk $f_j$ is said to \emph{decoded} by user $u_m$ if the user can obtain the corresponding
unit encoding vector $\mathbf{e}_j$ (possibly after Gaussian elimination) from its knowledge space
$\mathbf{K}_m(t)$.

An active user $u_m$ at timeslot $t$ is a user whose \ac{d.o.f.} satisfies $r_m(t) < T$.  
User $u_m$ is said to \emph{depart} the queue at time $t$ when the rank of its knowledge space becomes $T$.  
Throughout the rest of the paper, a user always refers to an active user who has not yet departed from the server's
buffer.

References \cite{DimRamWuSuh:11, FerMedSol:12} have assumed perfect scheduling by
the server, which is not assumed in our model.  Somewhat related to this issue is
the assumption  that the coefficients of a coded chunk are cycled or refreshed to ensure
innovative  chunks for every drive read.  Finally, to be able to apply queuing theoretical arguments
in  \cite{FerMedSol:12}, requests for different file chunks of the same content arrive randomly  and
independently  of other chunks at the server.  In that paper, the notion of users is abstracted
away, which we do not do here.

\vspace{-0.1cm}
\subsection{Performance Metrics}
\label{sec:metrics}
\vspace{-0.1cm}Let $\mathcal U_n(t) \subset \mathcal U_A$ be the subset of \emph{targeted} users who receive an innovative \ac{d.o.f.} from the
broadcast of chunk $x(t)$ at timeslot $t$.  We define three throughput metrics in order of
strongest to weakest, which are equivalent to those used in cross-bar switch scheduling \cite{ChaCheHua:99}.  
\begin{definition}\label{def:optimal} \textit{(Throughput optimal)}
  A scheduling service is \emph{throughput optimal} if every service can guarantee $\mathcal U_n(t) = \mathcal U_A$. That is,
   $ r_m(t+1)= r_m(t)+1,  \forall u_{m} \in \mathcal{U}_A,  \forall t$. 
\end{definition}
Since system constraints may mean that throughput optimality is not feasible, we consider maximum and maximal throughput, which are in general
the best any scheduling scheme can do \emph{up to} or \emph{at} any timeslot based on constraints such as drive availability.
\begin{definition} \textit{(Maximum throughput)}
  A service scheme  achieves \textit{maximum throughput} if the total number of targeted users up to time $t$,
  denoted by $ \sum_{i=1}^{t} \vert \mathcal U_{n}(t) \vert$ is maximized, across all service
  schemes for a given data storage allocation.
\end{definition}
\begin{definition} \textit{(Maximal throughput)}
  A scheduling service achieves \textit{maximal throughput} if at each timeslot $t$, the
  number of targeted users $|\mathcal U_n(t)|$ is maximized, across all service
  schemes for a given data storage allocation. 
\end{definition}
Note that any service scheme that achieves maximal throughput is necessarily a greedy algorithm.  In
a given timeslot, active users that are not targeted by a scheduling scheme
are said to
be \emph{internally blocked}.  These users are not externally blocked as they are already in the
server's buffer, but are held up for service. The better the throughput of a scheduling
scheme, the lower its internal blocking probability will be.  Intuitively, a lower internal blocking
probability should lead to lower external  blocking probability as active users are flushed out of the
system faster.


\vspace{-0.15cm}

\section{Data Scheduling Schemes}
\label{sec:schemes}
We introduce the concept of a service \textit{leader} and considers two system types.
First, to develop intuition for our problem and to verify expectations, we consider systems in which drives never block, i.e., drives with infinite I/O access
bandwidth.  Second, we consider systems with  traffic-induced drive blocking, i.e., drives with finite
I/O access bandwidth.  In both systems, we propose service schemes for \ac{UCS} and \ac{NCS}.  Schemes presented in
this section can be formulated as  integer linear programs over content demand graphs, similar to those for cross-bar  switches
\cite{ChaCheHua:99,SunDebMed:07} and are omitted here.


\vspace{-0.2cm}
\subsection{Infinite I/O access bandwidth systems}
\label{sec:singleRead}

To verify expectations, consider a system in which drives have infinite I/O access bandwidth, so $ P_{b}^{D}=0$.  

\subsubsection{Uncoded Scheme}
\label{sec:schemeInfBW}

Consider \ac{UCS} and the scheme outlined in Algorithm \ref{alg:InfLeaderBasedScheme}.  
We introduce the following terminology for our leader-based scheme, which will also be used in
the finite I/O access bandwidth case. Let $\mathbf a_m(t)$ be a binary valued \emph{decoded chunk vector} of
length $T$ for user $u_m$ with elements $a_{m,j}(t)$.   If $a_{m,j}(t) = 0$  then user $u_m$ has
decoded file chunk $f_j$ and if $a_{m,j}(t) = 1$ then file chunk $f_j$ is yet to be decoded.
Upon arrival of user $u_m$'s file request, $a_{m,j}(t) = 1$ for all $1 \le j \le T$ and
upon departure $a_{m,j}(t') = 0$ for all $1 \le j \le T$.
\begin{itemize}
\item The \emph{leader user} $u_\ell$ at timeslot $t$ is the user with maximum knowledge space rank. That is,
\begin{align}
\label{eq:lmax}
  \ell = \mathrm{argmax}_{m:u_m\in \mathcal{U}_A} r_m(t) \, .
\end{align}
\item The \emph{earliest undecoded chunk} or simply \emph{min chunk} of user $u_m$ is the chunk for
which  $a_{m,j}(t) = 1$ and all $a_{m,k}(t) = 0$ for $k < j$. 
\item The \emph{earliest undecoded chunk of the leader} or simply \emph{min-max chunk} $f_{j^*}$ is
the chunk  for which $a_{\ell,j^*}(t) = 1$ and $a_{\ell,k}(t) = 0$ for $k < j^*$ for the leader user $u_\ell$.
\end{itemize}

See Fig.~\ref{fig:leaderSchemeEg} for an
example of the leader-based scheme in Algorithm \ref{alg:InfLeaderBasedScheme}.

\begin{algorithm}
\caption{Leader-based scheduling scheme for \ac{UCS} with infinite I/O access bandwidth.}
\label{alg:InfLeaderBasedScheme}
  \begin{algorithmic}[1]
\FOR{timeslot $t$}
\STATE Find the leading user $u_\ell$ among all active users in $ \mathcal U_A$, which has the
highest knowledge  space rank $r_\ell(t)$.
\STATE Find the leader's earliest undecoded chunk denoted by $f_{j^*}$.
\STATE Read $f_{j^*}$ from a drive in $\mathcal M^d_{j^*}$ and broadcast $x(t) = f_{j^*}$ to all active users.
\STATE All active users get to decode $f_{j^*}$ and the server updates  the decoded chunk list of
all active users.  That is, $a_{m,j^*}(t) = 0$ for all active users $u_m \in \mathcal U_A$.
\ENDFOR
  \end{algorithmic}
\end{algorithm}

\begin{figure}[tb]
  \centering
  \includegraphics[width=0.95\linewidth]{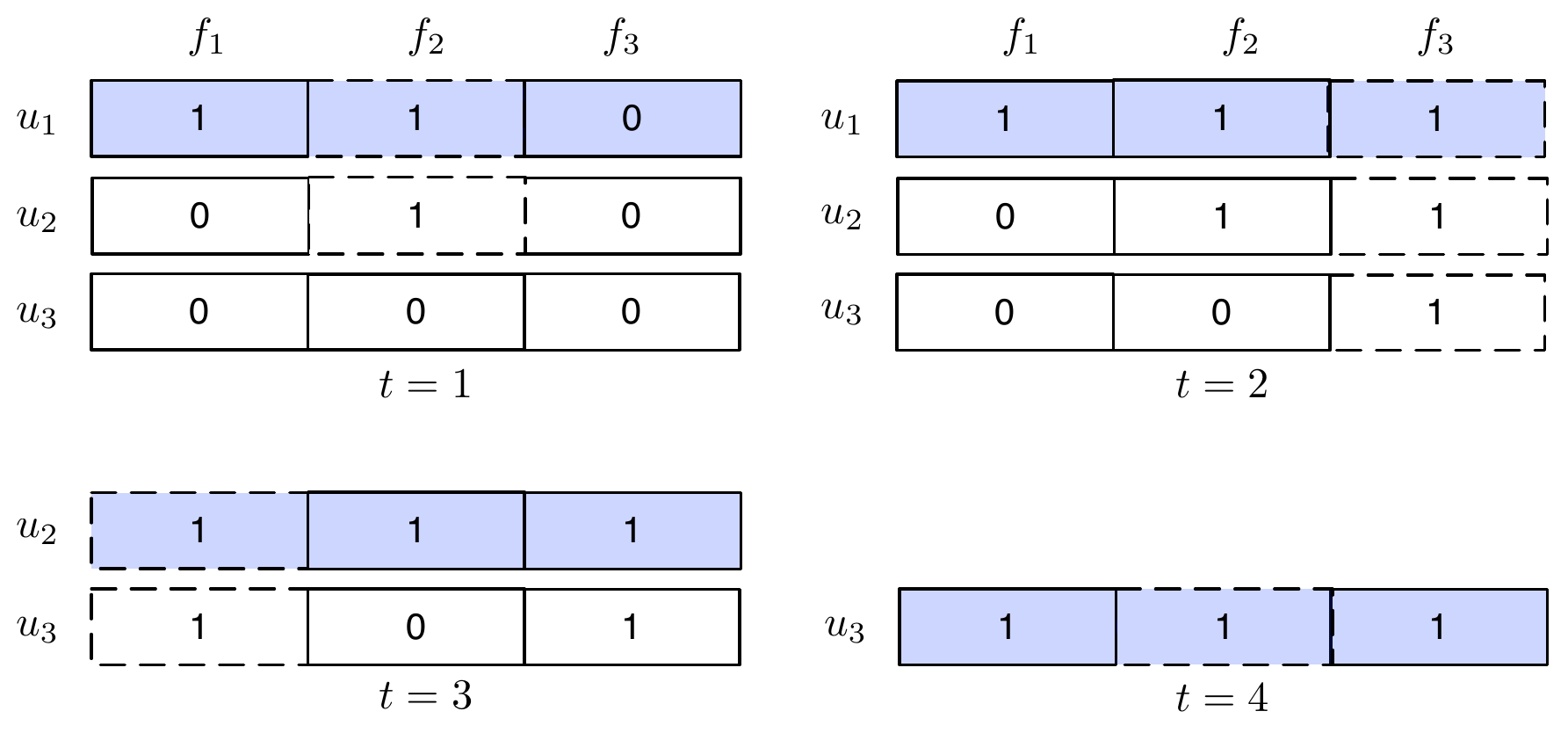}
  \caption{  Consider a system with $ T=3$, and the following example sequence of Algorithm \ref{alg:InfLeaderBasedScheme} showing the evolution of
    users' decoded chunk vector.  Users $u_1$ and $ u_{2}$ are already in the system with different
    demands when user $u_3$ arrives at $t=1$.
    During each timeslot $t$, $S$ transmits the earliest undecoded chunk of the shaded leader.}
  \label{fig:leaderSchemeEg}
\end{figure}

\subsubsection{Uncoded Scheme Analysis}
\label{sec:analysisInfBW}
Intuitively a system with infinite I/O access bandwidth and perfect communications will allow for
throughput optimal scheduling by $S$, since the scheme is without errors.  We now formalize that the leader-based scheme in Algorithm \ref{alg:InfLeaderBasedScheme} is
throughput optimal according to Definition \ref{def:optimal}.  

\begin{lem}
  The scheduling scheme of Algorithm \ref{alg:InfLeaderBasedScheme} is throughput optimal.
\end{lem}
\begin{proof}
We prove optimality by induction.

Base step: Consider an empty server queue.  When the first user arrives, it immediately become
the leader and the system services uncoded chunks sequentially starting from file chunk $ f_{1}$.
Therefore, in each timeslot this user will successfully receive a \ac{d.o.f.} so the scheme
is throughput optimal during this time.  

Inductive step:  Consider a throughput optimal scheme with a server queue comprising $m$ users in
$\mathcal{U}_A$  where all
users have received a \ac{d.o.f.} in all previous timeslots. The $(m+1)$th user arrives. Since
the knowledge space rank of the new $(m+1)$th user is zero, the leader remains unchanged.  Choose
uncoded chunk $x(t) = f_{j^*}$ corresponding to the leader as per Algorithm \ref{alg:InfLeaderBasedScheme}.  Then the
new user will also receive a \ac{d.o.f.} as it has received no chunks so far.   So all users
continue to receive a \ac{d.o.f.} in every timeslot and the scheme remains throughput optimal.  
\end{proof}

\begin{lem}
\label{lem:throughputToBlocking}
If a scheduling scheme is throughput optimal, then it also
minimizes the blocking probability across all feasible scheduling schemes.
\end{lem}

\begin{proof}
  A  throughput optimal scheme means that all users in $ \mathcal U_A$ receive an innovative
  \ac{d.o.f.} in each timeslot.  This means all users are serviced in $T$ timeslots after their arrival,  which is the
  minimum number possible because only one chunk can be broadcast per timeslot.  Hence, the service
  rate $ \mu = \vert \mathcal U \vert /T$ is at the maximum for a throughput optimal scheme.  
Given a fixed arrival rate $ \lambda$ and a fixed buffer size $ N$, the Erlang B blocking formula
monotonically decreases with increasing service rate $ \mu$.  Hence the
maximum service rate results in minimum blocking probability.  
\end{proof}

Applying Lemma \ref{lem:throughputToBlocking} to Algorithm \ref{alg:InfLeaderBasedScheme} shows that  it is also blocking probability optimal.

\begin{lem}
The blocking probability of a throughput optimal scheme is given by
\begin{align}
\label{eq:server_Pb_infiniteIO}
  P_{b}^{s} = \frac{ (\lambda T)^{N}/N! }{ \sum_{i=0}^{N} (\lambda T)^{i} / i! }  \, .
\end{align}
\end{lem}

\begin{proof}
  The arrival process for $S$ is a Poisson process.  Under a throughput optimal scheme, all
  users immediately begin being serviced upon arrival until a total of $N$ users are in the server
  buffer.  We can view each active user as being serviced by an individual service unit with
  deterministic service time $T$ timeslots.  Hence, the average service rate is $ 1/T$ for each
  server and $S$ is equivalent to an $ M/D/M/M$ queue, where
  $D$ is a deterministic service time.  The blocking probability for $S$ is then
  given by (\ref{eq:server_Pb_infiniteIO}).  
\end{proof}
  In a  system with infinite I/O access bandwidth all drives are always available for read. Then
  there is no need to  store more than one copy of each file chunk. That is, $|\mathcal M^d_{i}|= 1$
  for all  $f_i \in \mathcal F$ suffices for throughput optimality. 
\begin{rem}
Serving the earliest undecoded chunk of the leader is not essential for the optimality
of the algorithm.  Selecting any undecoded chunk by the leader will suffice.   However,
by serving undecoded chunks of the leader in a contiguous way, we promote better in-order delivery
to the application.  
\end{rem}

This verifies the intuitive result that \ac{NCS} does not provide benefit over \ac{UCS} in an
infinite I/O access bandwidth system.  Note that Algorithm \ref{alg:InfLeaderBasedScheme} can be
adjusted  to operate with coded storage via simple
modifications.   

\subsection{Finite I/O access bandwidth systems}
\label{sec:multipleRead}

In this subsection we consider systems with drives that can become busy owing to serving other
requests, i.e., drives with finite I/O access bandwidth for which $ P_{b}^{D} > 0$.  We still assume ideal chunk transport
medium with no erasures and broadcast capabilities to all active users, such as TCP
for multicast variants, Ethernet, or emulated broadcasting systems.

\subsubsection{Uncoded Scheme}\label{sec:schemeFiniteBW}

In the finite I/O case, the concept of leaders needs modification depending on what chunks
are available for access.  We then
distinguish between a \emph{true} leader and a \emph{temporary} leader in our modified scheduling
algorithm.  This is to  handle temporary
unavailability of drives that store undecoded chunks demanded by the true leader.  We modify
Algorithm \ref{alg:InfLeaderBasedScheme} to find  undecoded chunks of the true leader that
are available for read.   If no such undecoded chunk for the true leader is available, then we will
limit our search to the \emph{next} leading user and the undecoded chunks of that user, which by the
approach of the service scheme must have been all decoded by the excluded leader.  We continue until
we can find one user who is leading among the remaining users and for whom one of its undecoded
chunks is available for read.  The modified scheme operates as per Algorithm \ref{alg:modifiedScheme}.  

\begin{algorithm}
\caption{Leader-based scheduling scheme for \ac{UCS} with finite I/O access bandwidth.}
\label{alg:modifiedScheme}
  \begin{algorithmic}[1]
\FOR{timeslot $t$ }
    \STATE Obtain the drive availability vector $\bm{b}(t)$.
          \STATE Create a temporary list of active users, denoted by $\mathcal{U}_t$, and initialize it to all active users in the
      system: $\mathcal{U}_t = \mathcal{U}_A$.
          \STATE Find the leader from the temporary list of active users from (\ref{eq:lmax}).\label{step:leader}
          \STATE Find the leader's set of all undecoded chunks denoted by $\mathcal F^u_\ell \subset
      \mathcal F$. That is, $f_j \in \mathcal F_\ell^u \leftrightarrow a_{\ell,j}(t) = 1$.
          \STATE If there exists at least one available drive for at least one chunk in $\mathcal F^u_\ell$, then
      select one such chunk, denoted by $f_{j^*}$, and go to step \ref{step:serve}. Otherwise, remove the
      leader from temporary active users ($\mathcal{U}_t \leftarrow \mathcal{U}_t \setminus u_\ell$)
      and go to step  \ref{step:leader}.
          \STATE Read the chunk $f_{j^*}$ from one of the available drives in $\mathcal M^d_{j^*}$
          and broadcast $x(t) = f_{j^*}$  to all active users.\label{step:serve}
          \STATE All users decode the chunk $f_{j^*}$ and the server updates their
      decoded chunk list. That is, $a_{m,j^*}(t) = 0$ for all active users in $\mathcal{U}_t$. (Note
      that the excluded leading users have already decoded $f_{j^*}$ and hence at the end of this
      step $a_{m,j^*}(t) = 0$ for all users $u_m \in \mathcal U_A$).
\ENDFOR
        \end{algorithmic}
\end{algorithm}

\subsubsection{Coded Scheme}
The proposed scheme for \ac{NCS} finite I/O bandwidth systems is similar to Algorithm
\ref{alg:modifiedScheme} in terms of finding temporary leaders depending on drive availability. The
main difference with Algorithm \ref{alg:modifiedScheme} is the choice of the chunk for service: The
scheduler needs to keep track  of coded chunks so far received by the users.

For each timeslot $t$, we define a binary \emph{coded chunk reception} vector of size $H$ for user $u_m$, denoted by
$\bm{q}_m(t)$, as follows: $q_{m,r}(t) = 0$ if coded chunk $c_r$ has been so far received by user
$u_m$ and $q_{m,r}(t) = 1$ otherwise.  Algorithm \ref{alg:modifiedCodedScheme} describes the
scheme.  

\begin{algorithm}
\caption{Leader-based scheme for \ac{NCS} with finite I/O access bandwidth.}
\label{alg:modifiedCodedScheme}
  \begin{algorithmic}[1]
\FOR{timeslot $t$ }
    \STATE Obtain the drive availability vector $\bm{b}(t)$.
          \STATE Create a temporary list of active users, $\mathcal{U}_t = \mathcal{U}_A$.
          \STATE Find the leader from the temporary list of active users. \label{step:leader}
          \STATE Find the leader's set of all unreceived coded chunks denoted by $\mathcal
          C^u_\ell$. That is,  $c_r\in \mathcal C_\ell^u \leftrightarrow q_{\ell,r}(t) = 1$.
          \STATE If there exists at least one available coded chunk in $\mathcal C^u_\ell$ for read, then
      select one such chunk, denoted by $c_{r^*}$, and go to step \ref{step:serve}. Otherwise, remove the
      leader from temporary active users ($\mathcal{U}_t \leftarrow \mathcal{U}_t \setminus u_\ell$)
      and go to  step \ref{step:leader}.
          \STATE Read the chunk $c_{r^*}$ from its corresponding drive and broadcast $x(t) =
          c_{r^*}$ to all active users. \label{step:serve}
          \STATE Update $q_{m,r^*}(t) = 0$ for all active users $u_m \in \mathcal U_A$.
\ENDFOR
        \end{algorithmic}
\end{algorithm}

\subsubsection{Schemes Analysis and Comparison}

\begin{figure}[tb]
  \centering
  \includegraphics[width=\linewidth]{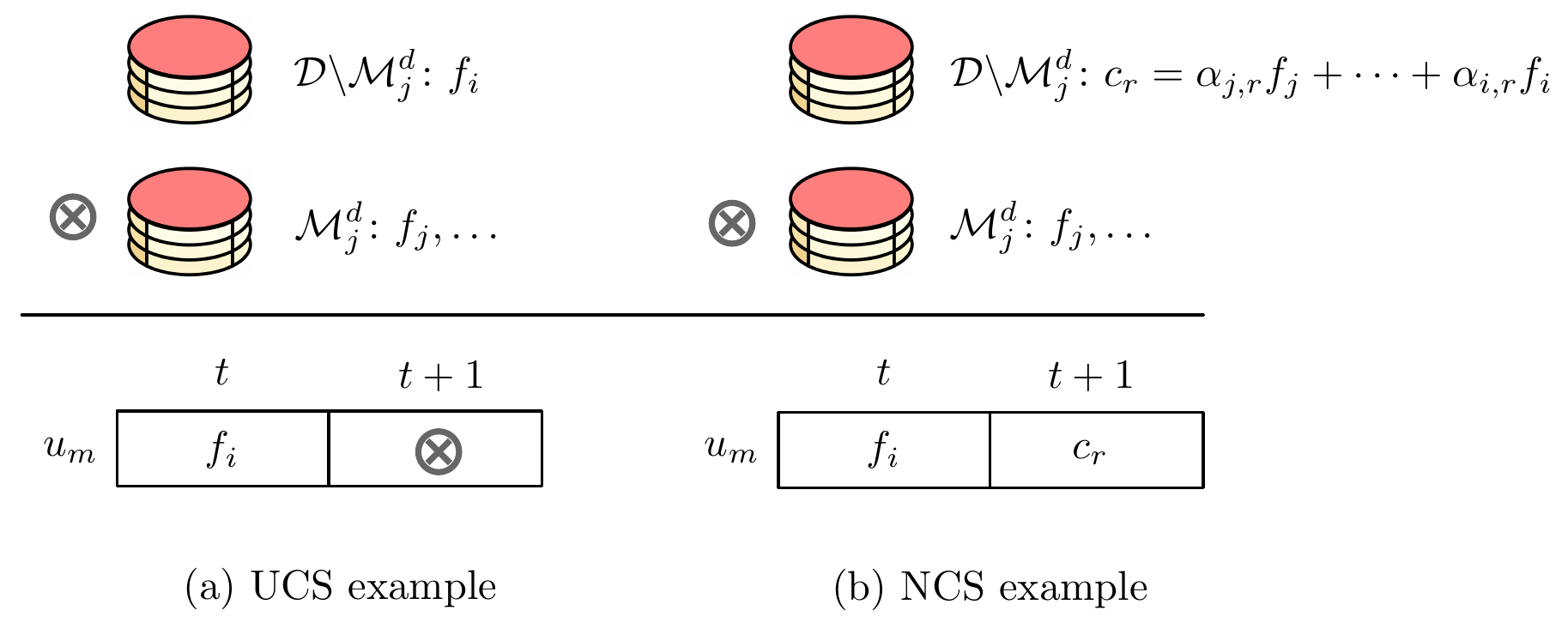}
  \caption{Given an instance of \ac{UCS}, an example construction \ac{NCS} can improve throughput compared to a \ac{UCS}. }
  \label{fig:precodingToyExample}\vspace{-2mm}
\end{figure}

In a finite I/O storage system neither \ac{UCS} nor \ac{NCS} can guarantee throughput optimality,
since we can always find a drive unavailability pattern with non-zero probability of occurring that
would block at least one user (for instance, consider the simple case when all drives are blocked
in the same timeslot).  We now show simple proofs showing that while both uncoded schedule of Algorithm
\ref{alg:modifiedScheme} and coded schedule of Algorithm \ref{alg:modifiedCodedScheme} achieve
maximal throughput across their respective data storage formats, that the number of targeted users
using \ac{NCS} with maximal throughput scheduling is at least as high as that in the \ac{UCS} system.

\begin{lem}
\label{lem: MaximalThroughput}
Algorithms \ref{alg:modifiedScheme} and \ref{alg:modifiedCodedScheme} achieve maximal throughput
across their respective data storage formats.
\end{lem}

\begin{proof}   
The scheme of Algorithm \ref{alg:modifiedScheme} (Algorithm \ref{alg:modifiedCodedScheme}) identifies a leading user with maximum rank with
available file chunk(s) for read.  All other users with
smaller or equal ranks will also receive a \ac{d.o.f.} since no chunks exist that non-leader users
have decoded (received) but the leader has not.  Consequently, at any given time, the number of
serviced users with a \ac{d.o.f.} is maximized subject to instantaneous drive availability given the storage format. Therefore, both algorithms achieve maximal throughput. 
\end{proof}
 
 \begin{lem}
\label{lem: PrecodingHelps}
The number of targeted users $|\mathcal U_n(t)|$ in a finite I/O storage system using \ac{NCS}
with maximal throughput scheduling can be at least as high as that in a \ac{UCS}
system with maximal throughput scheduling.
\end{lem}

\begin{proof}
Given any instance of \ac{UCS}, we need to show (1) that no drive blocking patterns exist where the number of targeted users $|\mathcal U_n(t)|$ is higher than that in all instances of \ac{NCS}, and (2) that there exist drive blocking patterns for which $|\mathcal U_n(t)|$ in \ac{NCS} 
is higher than that in the \ac{UCS} instance.  

For (1), consider an instance of \ac{NCS} which is constructed as follows.  Each coded chunk $ c_{r}$ stored on $ D_i$ is a linear combination of the
uncoded chunks stored on $D_{i}$ in the \ac{UCS} instance.  Under this
scenario, given linear independence from earlier chunks, if any read from $D_i$ in the \ac{UCS} instance can target $|\mathcal U_n(t)|$ users, it is clear that a read from $D_i$ in the \ac{NCS} counterpart can also provide a new \ac{d.o.f.} to at least the same number of users.     
For (2), we proceed by counterexample.  We can always consider a single active user $u_m$ with $r_{m}(t) = T-1$ under the \ac{UCS}
instance.  See Fig.~\ref{fig:precodingToyExample} for a toy-example, when $r_{1}(2) = 1$ and the
only missing chunk of user $u_1$ is $f_j$.  Assume that all drives in
$\mathcal M_{j}^{d}$ are
blocked during timeslot $ t+1$.  For the \ac{UCS} system, $u_m$ cannot be targeted so $\vert
\mathcal U_{n}(t_+1) \vert =0$.  However, in the \ac{NCS} instance of the system, although $\mathcal M_{j}^{d}$ is
blocked, any unseen coded chunks with $ \alpha_{j,r} \neq 0$ stored on drives in $ \mathcal D
\backslash \mathcal M_j^d$ can still provide a new \ac{d.o.f.} to user $u_m$, so $\vert \mathcal U_{n}(t+1) \vert =1$.
\end{proof}

To further illustrate \ac{NCS}  improved blocking performance, we now focus on the internal true
leader's blocking probability.  First, consider the
following restricted \ac{UCS} file layout with replication and striping.  There are a total of $R = WT$ drives in the system. The uncoded system stores a single file chunk per drive, where drive $D_{(w-1)T+i}$ stores the $w$th copy of the $i$th file
chunk for $w = 1, \cdots, W$ and $i = 1, \cdots, T$. The coded system stores $H = R = WT$ coded file
chunks such as $c_r$, one on each drive $D_r$. Assuming that each drive becomes unavailable with
probability $P_b^D$ independently of other drives and previous timeslots, the following lemma gives
the leader internal blocking probability  in each system.

\begin{lem}
\label{lem:BlockingCompareFinite}
At timeslot $t$, the internal blocking probability of the true leader who has a knowledge space rank of $r_\ell(t)$ is given by
\begin{equation}\label{eq:block:coded:single}
P_{b}^c = (P_b^D)^{WT-r_\ell(t)}
\end{equation}
in an \ac{NCS} system and by
\begin{equation}\label{eq:block:uncoded:single}
P_{b}^u = (P_b^D)^{WT-Wr_\ell(t)}
\end{equation}
in a \ac{UCS} system.
\end{lem}
\begin{proof}
If the leader has received $r_\ell(t)$ coded chunks up to time $t$, there remain only
$WT-r_\ell(t)$ useful drives for service and \eqref{eq:block:coded:single} follows. In the uncoded
system, if the leader has decoded $r_\ell(t)$ file chunks up to time
$t$, there remain only $WT-Wr_\ell(t)$ useful drives for service and \eqref{eq:block:uncoded:single}
follows. 
\end{proof}
For large $r_\ell(t)$ or $W$,  the improvement in leader blocking probability enabled by
coded storage can  become significant.   Next we consider regular $s$-striped storage systems \cite{FerMedSol:12} with a total of $R = Ws$ drives and $T/s$ file chunks in each stripe set which is
an integer. The following lemma gives the leader blocking
probability in uncoded and coded  systems.

\begin{lem}
\label{lem:BlockingCompareFinite2}
Assume that at timeslot $t$, the leader in the uncoded system has completely decoded $r$ out of $s$
stripe sets, where $r = 0, \cdots, s-1$, such that its knowledge space rank satisfies $rT/s\le
r_\ell(t) < (r+1)T/s$. Then, its  internal blocking  probability is given by
\begin{equation}\label{eq:block:uncoded:stripe}
P_{b}^u = (P_b^D)^{Ws-Wr}
\end{equation}
Now assume that in the coded system, the leader's knowledge space rank is also $r_\ell(t)$.  A
simple upper bound for the internal blocking probability $ P_{b,ub}^{c}$ is given by
\begin{equation}\label{eq:block:coded:stripe:worst}
(P_b^D)^{Ws-r} \le P_{b,ub}^{c} = (P_b^D)^{Ws-\lfloor \frac{r_\ell(t)}{T/s}\rfloor} < (P_b^D)^{Ws-(r+1)}
\end{equation}
And a simple lower bound $ P_{b,lb}^{c}$ is given by
\begin{equation}\label{eq:block:coded:stripe:best}
P_{b,lb}^{c} = (P_b^D)^{Ws-\max(0,r_\ell(t) - WT-Ws)}
\end{equation}
which will deviate from the best possible blocking probability of $(P_b^D)^{Ws}$ only when $W =1$ and $ T-s < r_\ell(t) < T$.
\end{lem}
\begin{proof}
In the uncoded system, if the leader has completely decoded
$r$ stripe sets up to time $t$, there only  remains $Ws-Wr$ useful drives for service and
\eqref{eq:block:uncoded:stripe}  follows. 

The worst case for the coded system occurs when during $r_\ell(t)$  previous services of the leader,
$\lfloor \frac{r_\ell(t)}{T/s}\rfloor$ out of $Ws$ available drives were completely read off and
hence are unavailable for further service,  in which case \eqref{eq:block:coded:stripe:worst}
follows. The bounds  are derived by using the inequalities $rT/s\le r_\ell(t) < (r+1)T/s$. 

The best case for the coded system occurs where all previous $r_\ell(t)$ services of the
leader were uniformly read across $Ws$ available drives. Therefore, one can verify that until the
leader's rank reaches $r_\ell(t) = Ws(T/s-1)+1 = WT-Ws+1$, none of the drives are completely read off
and are all available for service. Hence, we get, $P_{b,ub}^{c} = (P_b^D)^{Ws}$ for $r_\ell(t) <
WT-Ws+1$. After this point, the drives become sequentially unavailable and
$(P_b^D)^{Ws-\max(0,r_\ell(t) - WT-Ws)}$ follows. One can easily verify the last statement of the
lemma using $r_\ell(t) < T$, the assumption that $T/s$ is an integer and $s \le T/2$.  
\end{proof}


\begin{rem}
Lemma \ref{lem:BlockingCompareFinite2} demonstrates the importance of drive selection in Algorithms
\ref{alg:modifiedScheme} and \ref{alg:modifiedCodedScheme}, when more than one drive containing undecoded file chunks of the leader is
available for read.  
One can think about this as memory
in the system: Drive service units cease being helpful if all their content has been read.  When comparing different
variations of Algorithms \ref{alg:modifiedScheme} and \ref{alg:modifiedCodedScheme}, we expect that those which temporally spread reads
across drives to have better average throughput. 

\end{rem}

\section{Numerical Results}
\label{sec:numResults}

Using typical values found in various modern systems, we present Monte
Carlo simulation results comparing the performance of the proposed leader-based scheduling scheme
for \ac{UCS} and \ac{NCS} systems.  By using \eqref{eq:server_Pb_infiniteIO}, the analytical results for the
blocking probability of the proposed leader-based scheduling scheme for uncoded/coded storage with
infinite I/O access bandwidth are presented. For all simulations, we use a regular striped mapping
of file chunks onto  drives.

Fig.~\ref{fig:pbs_vs_N} illustrates the external blocking probability of the proposed scheduling scheme
for both uncoded and coded storage under drives' infinite and finite I/O access bandwidth conditions
versus server buffer size, $N$.  We see that when drives have  finite I/O access bandwidth, \ac{NCS} reduces system blocking probability over
\ac{UCS} and
that the gap tends to grow with increasing buffer size.

Fig.~\ref{fig:eta_pbs_vs_pbd} shows the average throughput and external blocking probability of the
proposed schemes for various drive internal blocking probabilities, $P_b^D$. As shown, 
throughput and external blocking probability are improved in \ac{NCS} compared with \ac{UCS} as
drives become more overwhelmed. In addition,  we see that a small 3\% improvement in throughput
renders a comparatively large improvement of 150\% in external blocking probability.

The internal blocking probability of the true leader versus its knowledge space rank in the uncoded
and coded storage for various drive internal blocking probabilities, $P_b^D$, is presented in
Fig.~\ref{fig:true_leader}. Here, the internal blocking probability of the true leader is
much lower in the coded system compared with the uncoded system as the leader's knowledge space rank
increases. This lower blocking probability is one factor explaining lower external blocking probability of the coded system compared to the uncoded system.
\begin{figure}[!t]
\centering
\hspace{-0.6cm}
\includegraphics[width=0.9\linewidth]{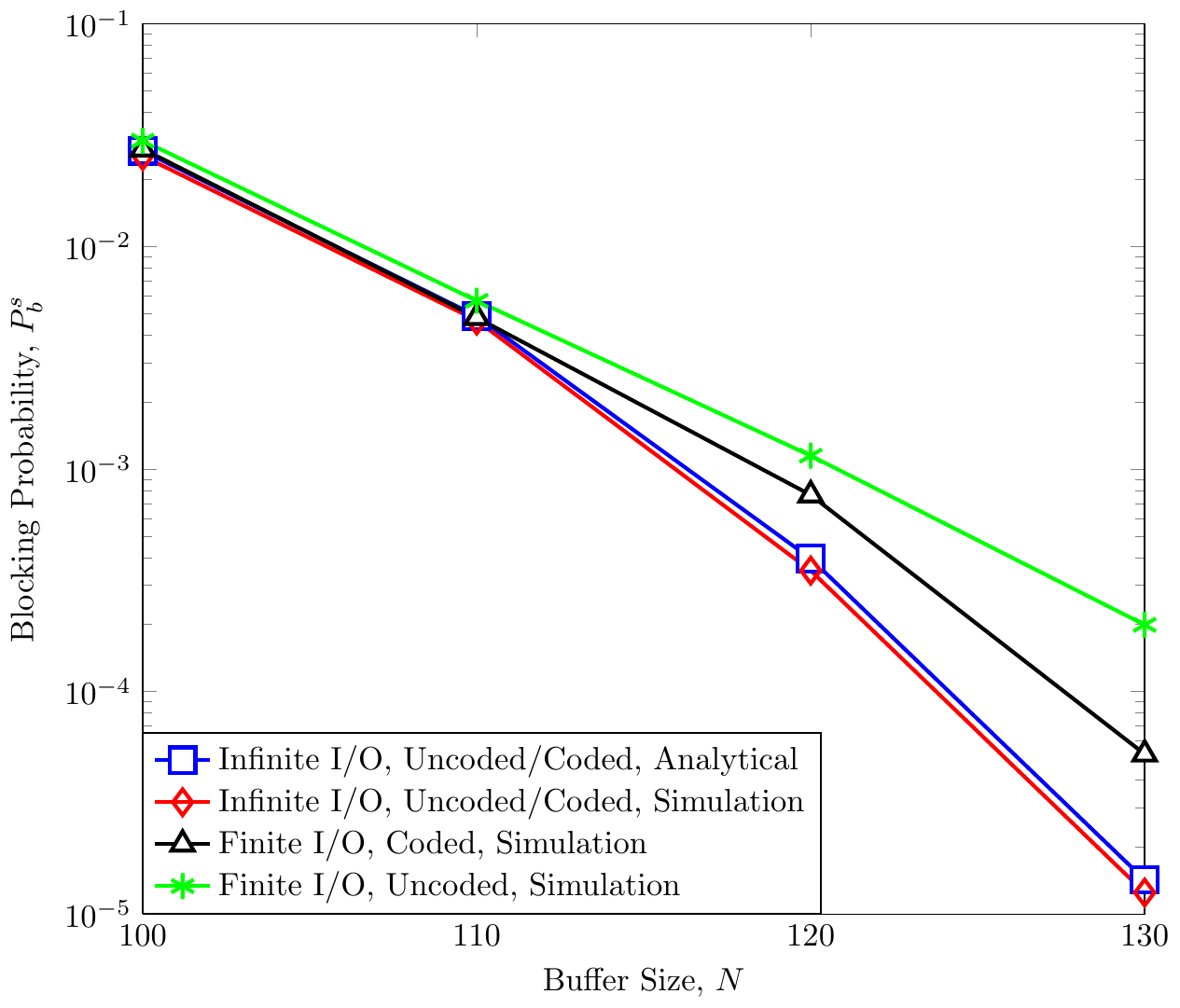}
\vspace{-0.2cm}
\caption{Blocking probability versus server buffer size $N$ for $\lambda= 0.9, T= 100, W= 2, R= 8, s= 4, P_b^D= 0.5$.}
\label{fig:pbs_vs_N}
\vspace{-0.2cm}\end{figure}

\begin{figure}[!t]
\centering
\hspace{0.1cm}
\includegraphics[width=0.95\linewidth]{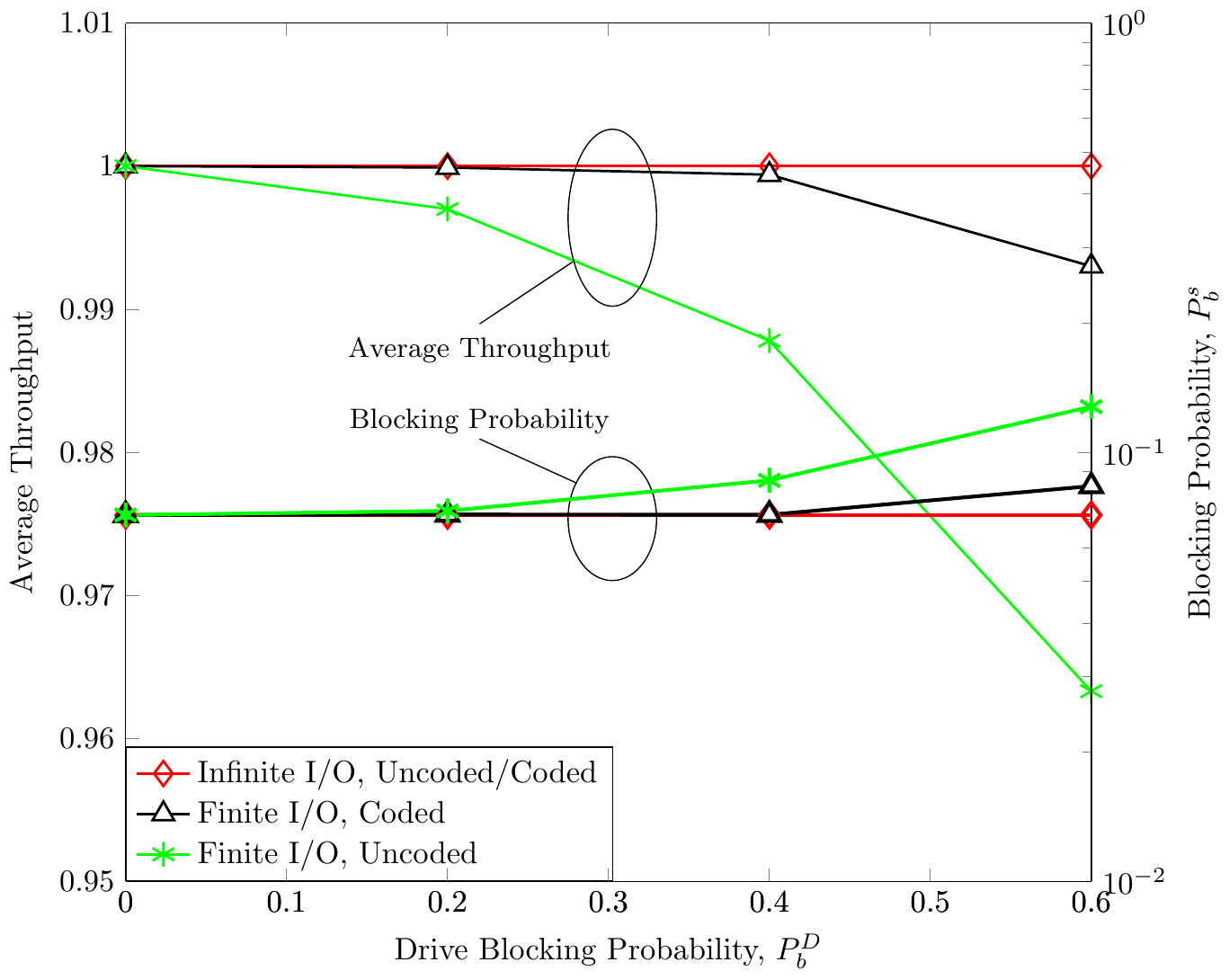}
\vspace{-0.2cm}
\caption{Average throughput and blocking probability versus $P_b^D$ for $\lambda=0.9, T= 8, W= 2, R= 8, s= 4$.}
\label{fig:eta_pbs_vs_pbd}
\vspace{-0.2cm}\end{figure}
\begin{figure}[!t]
\centering
\hspace{-0.6cm}
\includegraphics[width=0.9\linewidth]{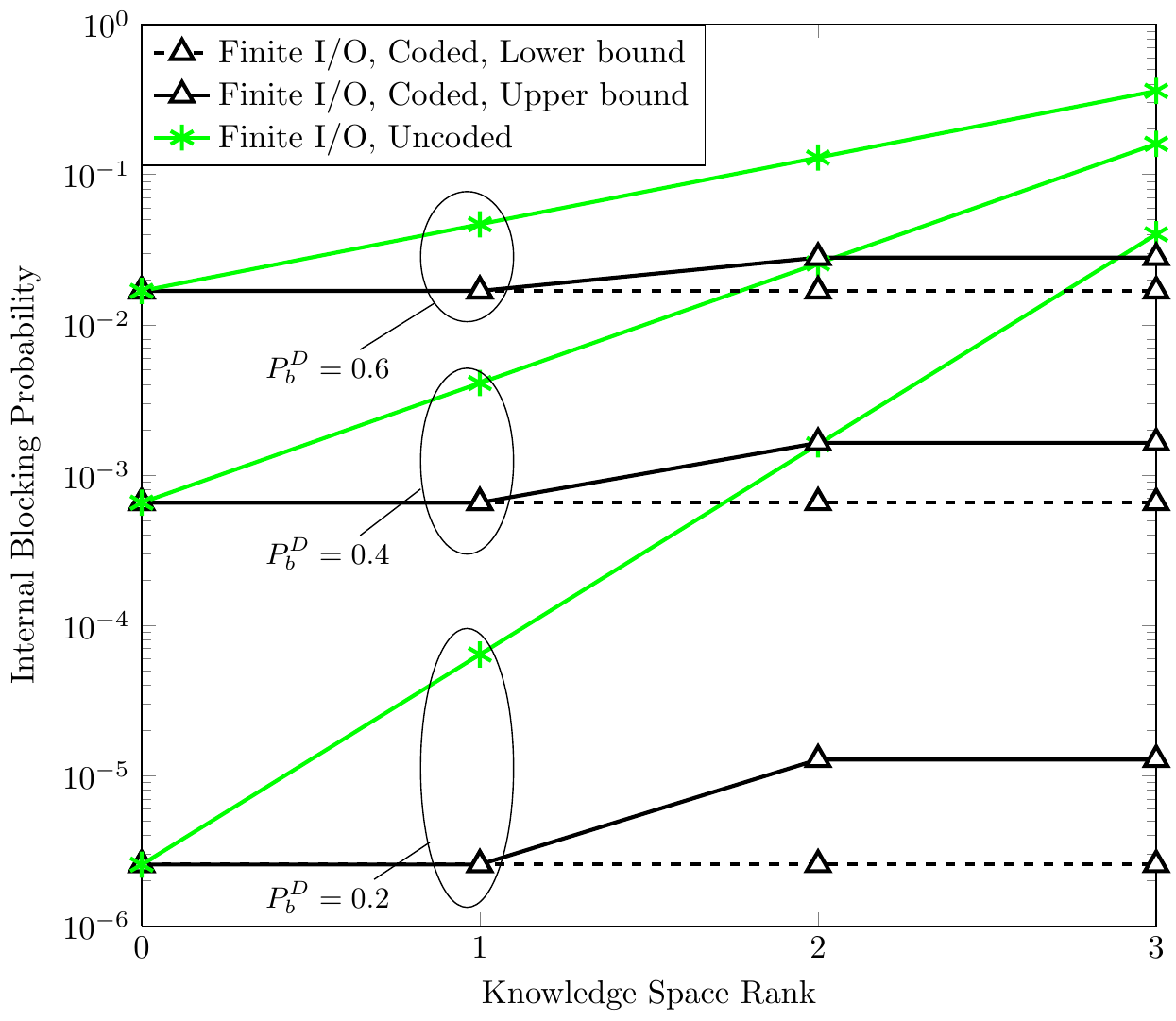}
\vspace{-0.2cm}
\caption{Blocking probability of the true leader versus drive blocking probability $P_b^D$ for $T= 8, W= 2, s= 4$.}
\label{fig:true_leader}
\vspace{-0.2cm}\end{figure}

\section{Conclusions}
\label{sec:conc}
In this paper, we introduced a novel and simple storage model for point-to-multipoint SANs and
investigated the impact of scheduling and content storage format on system blocking probability
and throughput in \ac{PMP} networks.  We proposed two intuitive drive access scheduling techniques for both
\ac{UCS} and \ac{NCS} systems, under infinite and finite I/O access bandwidth conditions.  In
finite I/O access networks, we
showed that \ac{NCS} scheduling flexibility
improves blocking probability and throughput  over \ac{UCS}.  Our
numerical evaluations and simulation results verify these advantages and can be used to guide future storage system design.


\begin{thebibliography}{10}
\providecommand{\url}[1]{#1}
\csname url@samestyle\endcsname
\providecommand{\newblock}{\relax}
\providecommand{\bibinfo}[2]{#2}
\providecommand{\BIBentrySTDinterwordspacing}{\spaceskip=0pt\relax}
\providecommand{\BIBentryALTinterwordstretchfactor}{4}
\providecommand{\BIBentryALTinterwordspacing}{\spaceskip=\fontdimen2\font plus
\BIBentryALTinterwordstretchfactor\fontdimen3\font minus
  \fontdimen4\font\relax}
\providecommand{\BIBforeignlanguage}[2]{{%
\expandafter\ifx\csname l@#1\endcsname\relax
\typeout{** WARNING: IEEEtran.bst: No hyphenation pattern has been}%
\typeout{** loaded for the language `#1'. Using the pattern for}%
\typeout{** the default language instead.}%
\else
\language=\csname l@#1\endcsname
\fi
#2}}
\providecommand{\BIBdecl}{\relax}
\BIBdecl

\bibitem{FerMedSol:12}
U.~J. Ferner, M.~Medard, and E.~Soljanin, ``Toward sustainable networking:
  {S}torage area networks with network coding,'' in \emph{Proc. Allerton Conf.
  on Commun., Control and Computing}, Champaign, IL, Oct. 2012.

\bibitem{DimRamWuSuh:11}
A.~G. Dimakis, K.~Ramchandran, Y.~Wu, and C.~Suh, ``A survey on network codes
  for distributed storage,'' \emph{Proc. {IEEE}}, vol.~99, no.~3, pp. 476--489,
  Mar. 2011.

\bibitem{ForLabPopStoTruBarGriQui:10}
\BIBentryALTinterwordspacing
D.~Ford, F.~Labelle, F.~I. Popovici, M.~Stokely, V.-A. Truong, L.~Barroso,
  C.~Grimes, and S.~Quinlan, ``Availability in globally distributed storage
  systems,'' in \emph{Proceedings of the 9th USENIX conference on Operating
  systems design and implementation}, ser. OSDI'10.\hskip 1em plus 0.5em minus
  0.4em\relax Berkeley, CA: USENIX Association, 2010, pp. 1--7. [Online].
  Available: \url{http://dl.acm.org/citation.cfm?id=1924943.1924948}
\BIBentrySTDinterwordspacing

\bibitem{ShaLeeRam:12}
N.~B. Shah, K.~Lee, and K.~Ramchandran, ``The {MDS} {Q}ueue: {A}nalysing
  latency performance of codes and redundant requests,'' \emph{CoRR,
  http://arxiv.org/abs/1211.5405}, 2012.

\bibitem{HuaPawHasRam:12}
L.~Huang, S.~Pawar, Z.~Hao, and K.~Ramchandran, ``Codes can reduce queueing
  delay in data centers,'' in \emph{Proc. IEEE Int. Symp. on Inf. Theory}, Jul.
  2012, pp. 2766--2770.

\bibitem{AndOwiSaxTha:93}
\BIBentryALTinterwordspacing
T.~E. Anderson, S.~S. Owicki, J.~B. Saxe, and C.~P. Thacker, ``High-speed
  switch scheduling for local-area networks,'' \emph{ACM Trans. Comput. Syst.},
  vol.~11, no.~4, pp. 319--352, Nov. 1993. [Online]. Available:
  \url{http://doi.acm.org/10.1145/161541.161736}
\BIBentrySTDinterwordspacing

\bibitem{MarBiaGiaLeoNer:03}
\BIBentryALTinterwordspacing
M.~A. Marsan, A.~Bianco, P.~Giaccone, E.~Leonardi, and F.~Neri, ``Multicast
  traffic in input-queued switches: optimal scheduling and maximum
  throughput,'' \emph{IEEE/ACM Trans. Netw.}, vol.~11, no.~3, pp. 465--477,
  Jun. 2003. [Online]. Available:
  \url{http://dx.doi.org/10.1109/TNET.2003.813048}
\BIBentrySTDinterwordspacing

\bibitem{PraMcKAhu:97}
B.~Prabhakar, N.~McKeown, and R.~Ahuja, ``Multicast scheduling for input-queued
  switches,'' vol.~15, no.~5, pp. 855--866, Jun. 1997.

\bibitem{YuReuBer:11}
H.~Yu, S.~Ruepp, and M.~S. Berger, ``Multi-level round-robin mulitcast
  scheduling with look-ahead mechanism,'' in \emph{Proc. IEEE Int. Conf. on
  Commun.}, Kyoto, Japan, Jun. 2011, pp. 1--5.

\bibitem{FenLi:11}
C.~Feng and B.~Li, \emph{Network coding: Fundamentals and applicaions},
  1st~ed.\hskip 1em plus 0.5em minus 0.4em\relax Academic Press, 2012, ch.
  Network coding for conten distribution and multimedia streaming in
  peer-to-peer networks.

\bibitem{RouSiv:97}
G.~N. Rouskas and V.~Sivaraman, ``Packet scheduling in broadcast {WDM} networks
  with arbitrary transceiver tuning latencies,'' \emph{{IEEE/ACM} Trans.
  Netw.}, vol.~5, no.~3, pp. 359--370, Jun. 1997.

\bibitem{AskFraZdo:01}
D.~Aksoy, M.~J. Franklin, and S.~Zdonik, ``Data staging for on-demand
  broadcast,'' in \emph{{P}roc. 27th {VLDB} {C}onf.}, Roma, Italy 2001.

\bibitem{Hu:01}
A.~Hu, ``Video-on-demand broadcasting protocols: {A} comprehensive study,'' in
  \emph{Proc. IEEE Conf. on Computer Commun.}, Anchorage, {AK}, Apr. 2001, pp.
  508--517.

\bibitem{HoMedKoeEffShiKar:06}
T.~Ho, R.~Koetter, M.~M\'{e}dard, M.~Effros, J.~Shi, and D.~Karger, ``A random
  linear network coding approach to multicast,'' \emph{{IEEE} Trans. Inf.
  Theory}, vol.~52, no.~10, pp. 4413--4430, Oct. 2006.

\bibitem{FerLonPedVolMed:13}
U.~J. Ferner, Q.~Long, M.~Pedroso, L.~Voloch, and M.~M\'{e}dard, ``Building a
  network coded storage testbed for data center energy reduction,'' in
  \emph{Proc. {IEEE} {SustainIT}}, Polermo, Italy, Oct. 2013.

\bibitem{ChaCheHua:99}
C.-S. Chang, W.-J. Chen, and H.-Y. Huang, ``On service guarantees for
  input-buffered crossbar switches: a capacity decomposition approach by
  {B}irkhoff and von {N}eumann,'' in \emph{Proc. {IWQoS}}, 1999, pp. 79--86.

\bibitem{SunDebMed:07}
J.~Sundararajan, S.~Deb, and M.~M\'{e}dard, ``Extending the {B}irkhoff-von
  {N}eumann switching strategy for multicast - on the use of optical splitting
  in switches,'' \emph{{IEEE} J. Sel. Areas Commun.}, vol.~25, pp.
  36--50, 2007.
  \bibitem{KimSunMedEryKot:11}
M.~Kim, J.~K. Sundararajan, M.~M\'{e}dard, A.~Eryilmaz, and R.~Kotter,
  ``Network coding in a multicast switch,'' \emph{{IEEE} Trans. Inf. Theory},
  vol.~57, no.~1, pp. 436--460, 2011.

\end{thebibliography}


\end{document}